\documentclass[11pt]{article}

\usepackage[numbers,sort&compress]{natbib}
\usepackage[T1]{fontenc}
\usepackage{graphicx}
\usepackage[utf8]{inputenc}
\usepackage{amsmath,mathrsfs,mathtools,amsthm}
\usepackage{amsmath,amsfonts,amssymb,epsfig}
\usepackage[usenames,dvipsnames]{color}

\usepackage{graphicx}
\usepackage{epstopdf}
\usepackage{yfonts}
\usepackage[T1]{fontenc}
\usepackage[all]{xy}
\usepackage{times}
\usepackage{hyperref}
\hypersetup{
    colorlinks=true,
    linkcolor=blue,
    citecolor=blue,
    urlcolor=blue
}

\usepackage{caption}
\usepackage[linesnumbered,ruled,vlined]{algorithm2e}

\newtheorem{Theorem}{Theorem}[section]
\newtheorem{Definition}[Theorem]{Definition}

\newtheorem{Corollary}[Theorem]{Corollary}
\newtheorem{Example}[Theorem]{Example}
\newtheorem{Remark}[Theorem]{Remark}

\newtheorem{Definition and Notation}[Theorem]{Definition and Notation}
\newtheorem{Lemma}[Theorem]{Lemma}

\begin{document}

\author{
  Maryam Bajalan\thanks{Institute of Mathematics and Informatics, Bulgarian Academy of Sciences, Bl. 8, Acad. G. Bonchev Str., 1113, Sofia, Bulgaria. Email: maryam.bajalan@math.bas.bg.} \and
  Peter Boyvalenkov\thanks{Institute of Mathematics and Informatics, Bulgarian Academy of Sciences, Bl. 8, Acad. G. Bonchev Str., 1113, Sofia, Bulgaria. Email: peter@math.bas.bg.}\and
  Ferruh Özbudak\thanks{Faculty of Engineering and Natural Sciences, Sabancı University, 34956, Istanbul, Turkey.  Email: ferruh.ozbudak@sabanciuniv.edu.}
  }

\title{Some structural properties of mixed orthogonal arrays and their irredundancy}

\date{}

\maketitle
\begin{abstract}
Mixed (asymmetric) orthogonal arrays (MOAs) generalize classical orthogonal arrays by allowing columns over different alphabets. However, their study requires very different structural tools than those used for symmetric orthogonal arrays (OAs), since several key features of the symmetric setting are no longer available in the mixed case, including Euclidean duality, a unique global index, and certain classical bounds. In this paper, we establish three structural results for mixed orthogonal arrays. First, we prove a Singleton-type upper bound and obtain a characterization of MDS and almost-MDS mixed orthogonal arrays. Second, we introduce a trace duality for $\mathbb{F}_q$-linear MOAs over 
$\prod_{i=1}^{s} \mathbb{F}_{q^{n_i}}$ and establish a correspondence with $\mathbb{F}_q$-linear error-block codes that determines the strength of the MOA via the dual distance of the associated error-block code. Finally, we develop a structural theory of irredundant mixed orthogonal arrays (IrMOAs), motivated by their role in the construction of $t$-uniform and absolutely maximally entangled (AME) quantum states. In the extremal case $t=\lfloor s/2\rfloor$, we prove that $\mathbb{F}_q$-linear IrMOAs with minimum index $1$ (yielding AME states of minimal support) are equivalent to $\mathbb{F}_q$-linear error-block MDS codes.
\end{abstract}
\begin{quotation}
\noindent \textbf{Keywords:} Mixed orthogonal arrays; Irredundant mixed orthogonal arrays; Error-block codes; MDS codes; Singleton-type bound; $t$-uniform states; Absolutely maximally entangled states.\\
\textbf{MSC(2020):} 05B15; 94B05; 81P40.
\end{quotation}

%=====================================================================
%Section: Introduction
%=====================================================================

\section{Introduction}
Mixed (asymmetric) orthogonal arrays (MOAs) are orthogonal arrays in which different columns take values from different alphabets, in contrast to symmetric orthogonal arrays (OAs), where all columns share a common alphabet. Early work on such arrays dates back to Addelman and Kempthorne \cite{Addelman_1962}, and the notion of asymmetrical orthogonal arrays was later formalized by Rao \cite{Rao1973}. Over the past decades, substantial progress has been made in constructing MOAs \cite{SuenDasDey2001,ZhangDengDey2017,ZhangZongDey2015,SuenDey2003,ZhangDuanYang2024,Suen2007,LinPang2025,LiZhuPang2025, NiuChenGaoPang2025, PangXuChenWang2018,PangEtAl2021}. Because the columns of MOAs take values from different alphabets, there are several differences between MOAs and OAs. For instance, the usual Euclidean inner product is not well-defined. Moreover, projections onto different sets of columns may have different cardinalities.  In addition, the indices associated with column subsets of an MOA depend on the chosen columns and therefore are not intrinsic parameters of the array. Consequently, the study of MOAs requires new tools that differ from those used for OAs.

The main contributions of this paper can be summarized as follows.

\begin{itemize}
    \item We establish a Singleton-type bound for mixed orthogonal arrays and derive a structural characterization of MDS and almost-MDS MOAs (Theorem \ref{Singleton}). 

{\it Explanation and Motivation}. A central problem in the theory of orthogonal arrays is to relate the (extremal) size of an array with its remaining parameters. The Rao and Singleton bounds are discussed in detail in \cite{HSS2012} (see Chapters 2 and 4 in that book). Orthogonal arrays of index $\lambda=1$ are the same as MDS codes. The original Rao and Singleton bounds were proved in 1947 \cite{Rao1947} and 1964 \cite{Singleton1964}, respectively. 

    \item We introduce a trace-based duality $\perp_{\mathrm{Tr}}$ for $\mathbb{F}_q$-linear MOAs over $\prod_{i=1}^{s}\mathbb{F}_{q^{n_i}}$, and construct an isomorphism $\rho$ to $\mathbb{F}_q$-linear error-block codes that preserves both distance and duality. Within this framework, an $\mathbb{F}_q$-linear MOA of strength $t$ has trace-dual distance at least $t+1$, with equality when the strength is exactly $t$ (Theorem~\ref{necessary-condition}). Furthermore, if $C$ is an $\mathbb{F}_q$-linear error-block code, then $\rho^{-1}(C)$ is an $\mathbb{F}_q$-linear MOA of strength $d_\pi(C^\perp)-1$ (Theorem~\ref{th:so}). This establishes a structural equivalence between linear MOAs and linear error-block codes.

{\it Explanation and Motivation}. In the symmetric setting, the codewords of a linear code form a linear orthogonal array whose strength is determined precisely by the Euclidean dual distance of the code (see \cite[Theorem~4.6]{HSS2012}). This connection was first observed by Kempthorne \cite{Kempthorne1947} and Bose \cite{Bose1961}, and later developed by Delsarte \cite{Delsarte1973}. However, this connection cannot be extended directly to the mixed setting, since the usual Euclidean inner product is not well-defined for MOAs.

    \item We define the minimum index $\lambda_{\min}$ of MOAs. We then extend the correspondence between a special class of IrOAs and AME states of minimal support in homogeneous systems to the heterogeneous setting. Theorem~\ref{cor:extreme_cases} shows that when $t=\lfloor s/2\rfloor$, IrMOAs with $M=\prod_{i=1}^{t} q^{n_i}$ (equivalently, with minimum index $\lambda_{\min}=1$) are equivalent to $\mathbb{F}_q$-linear error-block MDS codes. Consequently, this class of IrMOAs generates AME states of minimal support in heterogeneous systems. Moreover, Theorem \ref{latest} gives a method for constructing $\mathbb{F}_q$-linear IrMOAs from $\mathbb{F}_q$-linear error-block codes.
    
{\it Explanation and Motivation}. The notion of irredundant mixed orthogonal arrays (IrMOAs) originates in quantum information theory through the concept of $t$-uniform states. Let $H_q=\{0,1,\ldots,q-1\}$. An $s$-partite pure state in the heterogeneous system $\mathbb{C}^{q_1}\otimes\cdots\otimes\mathbb{C}^{q_s}$ can be written as $|\psi\rangle=\sum_{\mathbf{u}\in H_{q_1}\times\cdots\times H_{q_s}} a_{\mathbf{u}}\,|\mathbf{u}\rangle,$
where $a_{\mathbf{u}}\in\mathbb{C}$ and $|\mathbf{u}\rangle$ denotes the standard basis vector indexed by $\mathbf{u}$. 
When all $q_i$ are equal, the system is called homogeneous. The support of $|\psi\rangle$ is the number of non-zero coefficients in this expansion \cite{GALRZ2015}. 
The state is called $t$-uniform if every reduction to any subset of $t$ parties is maximally mixed \cite{GZ2014,S2004}. 
In the extremal case $t=\lfloor s/2\rfloor$, the state is called absolutely maximally entangled (AME). 
Such states play an important role in quantum information tasks such as quantum secret sharing and quantum error correction \cite{GALRZ2015,HGS2017,AlsinaRazavi2021,SSCZ2022}. 
In the homogeneous case, an AME state has support at least $q^{\lfloor s/2\rfloor}$ \cite[Proposition~1]{B2019}, and it is said to have minimal support when this bound is attained \cite{GALRZ2015}. Irredundant orthogonal arrays (IrOAs) and irredundant mixed orthogonal arrays (IrMOAs) were introduced by Goyeneche et al. \cite{GBZ2016, GZ2014} as a combinatorial framework for constructing $t$-uniform states in homogeneous and heterogeneous systems (see also \cite{Bajalan_2025,GALRZ2015,HGS2017,LiWang2019,SSCZ2022,ShenChen2021,ShiNingZhaoZhang2024,FJXY2023,ChenJiLei,PangZhangFeiZheng2021}). 
When $t=\lfloor s/2\rfloor$, IrOAs with $M=q^t$ (equivalently, with index $\lambda=1$) are equivalent to classical MDS codes \cite{B2019,GALRZ2015}. Consequently, this class of IrOAs constructs AME states of minimal support in homogeneous systems. 
\end{itemize}

The paper is organized as follows. Section 2 reviews $\mathbb{F}_q$-linear error-block codes of length $n$ and type $\pi$, including the $\pi$-distance and the block structure of generator and parity-check matrices. Section 3 introduces mixed orthogonal arrays $\mathrm{MOA}(M, s, (q^{n_1}, \dots, q^{n_s}), t)$ over the alphabet $\mathbb{F}_{q^{n_1}} \times \dots \times \mathbb{F}_{q^{n_s}}$, establishes a generalized Singleton-type bound, and defines the minimum index $\lambda_{\min}$ to characterize MDS and almost-MDS arrays. Section 4 develops the trace duality for MOAs and relates it to the Euclidean dual of the associated error-block code via an isomorphism $\rho$ that preserves both distance and duality. Section 5 proves the correspondence between MOA strength and dual distance and derives constructions of $\mathbb{F}_q$-linear MOAs from error-block codes. Finally, Section 6 investigates IrMOAs, proves that $\mathbb{F}_q$-linear IrMOAs with minimum index $1$, which yield AME states of minimal support, are equivalent to $\mathbb{F}_q$-linear error-block MDS codes, and presents criteria under which a linear error-block code and its dual generate IrMOAs.

%=====================================================================
%Section: Preliminaries
%=====================================================================

\section{Preliminaries}
In this section, we recall the definitions and some basic properties of $\mathbb{F}_q$-linear error-block codes. For further details, we refer the reader to \cite{Feng2006, Ling2007} and the references therein.

A partition of the integer $n$ is an expression of the form
$$
n = l_1 m_1 + l_2 m_2 + \dots + l_t m_t,
$$
where the integers satisfy $m_1 > m_2 > \dots > m_t \ge 1$ and $l_j \ge 1$ for all $1\le j\le t$. Here $m_j$ denotes the size of each block of type $j$, and $l_j$ indicates how many such blocks occur. The total number of blocks is therefore $s = l_1 + l_2 + \dots + l_t$. 
For convenience, define integers $n_1, n_2, \dots, n_s$ so that
\begin{align*}
n_1 &= \dots = n_{l_1} = m_1, \\
n_{l_1+1} &= \dots = n_{l_1+l_2} = m_2, \\
\vdots &\\
n_{l_1+\dots+l_{t-1}+1} &= \dots = n_s = m_t.
\end{align*}

This means that the sequence $(n_1, n_2, \dots, n_s)$ lists the block sizes $m_j$ in decreasing order as follows:
$$
\overbrace{%
\begin{array}{cccr}
m_1 & m_1 & \dots & m_1 \\
\downarrow & \downarrow & & \downarrow \\
n_1 & n_2 & \dots & n_{l_1}
\end{array}
}^{l_1}
\,\,
\overbrace{%
\begin{array}{cccr}
m_2 & m_2 & \dots & m_2 \\
\downarrow & \downarrow & & \downarrow \\
n_{l_1+1} & n_{l_1+2} & \dots & n_{l_1+l_2}
\end{array}
}^{l_2}
\,\,
\dots
\,\,
\overbrace{%
\begin{array}{cccr}
m_t & m_t & \dots & m_t \\
\downarrow & \downarrow & & \downarrow \\
n_{s-l_t+1} & n_{s-l_t+2} & \dots & n_{s} 
\end{array}
}^{l_t}.
$$
Thus, we can express
$$n = n_1 + n_2 + \cdots + n_s = l_1 m_1 + l_2 m_2 + \cdots + l_t m_t,$$
where $n_1 \ge n_2 \ge \dots \ge n_s$ and $s \ge 1$.
We represent this partition of $n$ either as
$$\pi = [m_1]^{l_1}[m_2]^{l_2}\cdots[m_t]^{l_t}
\quad \text{or equivalently as} \quad
\pi = [n_1][n_2]\cdots[n_s].$$
Let $q = p^{r}$ be a prime power, and let $\mathbb{F}_q$ denote the finite field with $q$ elements.
For the partition $\pi = [n_1][n_2]\cdots[n_s]$, the vector space $\mathbb{F}_q^n$ decomposes as
\begin{equation}
\mathbb{F}_q^n
= \mathbb{F}_q^{n_1} \times \mathbb{F}_q^{n_2} \times \cdots \times \mathbb{F}_q^{n_s}
= V_1 \times V_2 \times \cdots \times V_s,
\end{equation}
where each component $V_i$ denotes the subspace $\mathbb{F}_q^{n_i}$.

\begin{Definition}\label{ma18}
Let $\pi = [n_1][n_2]\cdots[n_s]$ be the partition of $n$, and let
$$
\boldsymbol{u} = (\boldsymbol{u}_1, \boldsymbol{u}_2, \dots, \boldsymbol{u}_s) \in \mathbb{F}_q^n = V_1 \times V_2 \times \cdots \times V_s,
$$
where each block $\boldsymbol{u}_i$ belongs to the corresponding subspace $V_i = \mathbb{F}_q^{n_i}$.
The $\pi$-weight of $\boldsymbol{u}$, denoted by $w_\pi(\boldsymbol{u})$, is defined as the number of non-zero blocks in $\boldsymbol{u}$, i.e.,  $w_\pi(\boldsymbol{u}) = \#\{ i \;:\, \boldsymbol{u}_i \ne 0 \}$.
\end{Definition}

\begin{Example}
 Let $n = 9$ and consider the partition $\pi = [3]^2[1]^3.$
We have
\begin{itemize}
    \item  $s = l_1 + l_2 = 2 + 3 = 5$;
    \item $ n_1 = n_2 = m_1 = 3, \quad n_3 = n_4 = n_5 = m_2 = 1. $   
\end{itemize}
Hence $\pi = [n_1][n_2][n_3][n_4][n_5]=[3][3][1][1][1].$ Thus
$$
\mathbb{F}_q^9
= \mathbb{F}_q^{3} \times \mathbb{F}_q^{3} \times \mathbb{F}_q^{1} \times \mathbb{F}_q^{1} \times \mathbb{F}_q^{1}
= V_1 \times V_2 \times V_3 \times V_4 \times V_5,
$$
which corresponds to the partition $\pi=[3][3][1][1][1]$. Consider the vectors
$$\boldsymbol{u} = \Big((1,0,0), (0,0,0), 0, 0, 0\Big),\quad \boldsymbol{v} = \Big((1,0,0), (0,1,0), 0, 0, 0\Big),\quad \boldsymbol{w} = \Big((1,1,1), (0,1,0), 1, 1, 1\Big).$$
The corresponding $\pi$-weights are $ w_\pi\big(\boldsymbol{u}\big)=1,$  $w_\pi\big( \boldsymbol{v}\big)=2$, and $w_\pi\big(\boldsymbol{w}\big)=5$.
\end{Example}

\begin{Definition}
An $\mathbb{F}_q$-linear error-block code of type $\pi=[n_1][n_2]\cdots[n_s]$ with parameters $[n, k, d_{\pi}]_q$ is an $\mathbb{F}_q$-linear subspace
$$
C \subseteq \mathbb{F}_q^n= \mathbb{F}_q^{n_1} \times \mathbb{F}_q^{n_2} \times \cdots \times \mathbb{F}_q^{n_s},
$$
where $k = \dim_{\mathbb{F}_q}(C)$, and the minimum $\pi$-distance $d_{\pi}$ of $C$ is defined by
$$
d_{\pi} = d_{\pi}(C):=\min \{ w_\pi(\boldsymbol{c}) \;:\; \boldsymbol{c} \in C, \boldsymbol{c} \ne 0 \}.
$$
\end{Definition}
An $\mathbb F_q$-linear error-block code of type $\pi = [1][1]\cdots[1]$ reduces to a classical $\mathbb{F}_q$-linear code. 

A generator matrix of an $\mathbb F_q$-linear error-block code $C$ of type $\pi = [n_1][n_2]\cdots[n_s]$ with parameters $[n,k,d_\pi(C)]_q$ is a $k \times n$ matrix of the form
\begin{align}\label{generator}
    G = \begin{bmatrix}
G_1 \mid G_2 \mid \dots \mid G_s
\end{bmatrix},
\end{align}
where each block $G_i$ is a $k \times n_i$ matrix corresponding to the $i$-th block of the partition and the rows of $G$ form an $\mathbb{F}_q$-basis of $C$. Similarly, a parity-check matrix of $C$ is
\begin{align}\label{parity}
    H = \begin{bmatrix}
H_1 \mid H_2 \mid \dots \mid H_s
\end{bmatrix},
\end{align}
where each $H_i$ is a matrix of size $(n - k) \times n_i$, and a vector $\boldsymbol{c} = (\boldsymbol{c}_1, \boldsymbol{c}_2, \dots, \boldsymbol{c}_s) \in \mathbb{F}_q^n$ is a codeword of $C$ if and only if
$$
H \boldsymbol{c}^T = 0
\quad \Longleftrightarrow \quad
\sum_{i=1}^{s} H_i \boldsymbol{c}_i^T = 0.
$$

The following lemma, which appears implicitly in \cite[Theorem 2.1]{Feng2006} without a proof, is stated and proved here as it will be used frequently in what follows.

\begin{Lemma}\label{ma14}
Let $C\subseteq\mathbb F_q^n$ be an $\mathbb F_q$-linear error-block code of type $\pi=[n_1]\cdots[n_s]$ with the minimum $\pi$-distance $d_{\pi}$ and with the parity-check matrix $H$ given in the form \eqref{parity}.
\begin{enumerate}
    \item If, for every choice of $t$ blocks of $H$, the union of their columns is an $\mathbb F_q$-linearly independent set, then $d_{\pi}(C)\ge t+1.$
    \item If there exists a set of $t+1$ blocks of $H$ whose union of columns is $\mathbb F_q$-linearly dependent, then $d_\pi(C)\le t+1.$
\end{enumerate}
\end{Lemma}

\begin{proof}
\begin{enumerate}
    \item Assume $d_\pi(C)\le t$. Thus there exists a non-zero vector $\boldsymbol{c}\in C$ whose $\pi$-weight is $r\le t$. Let $S=\{j_1,\dots,j_r\}$ be the indices of the non-zero blocks of $\boldsymbol{c}$.
    Thus $\boldsymbol{c}_i=0$ for all $i\notin S$, and $\boldsymbol{c}_{j_k}\neq 0$ for all $k\in \{1,\dots,r\}$. Since $\boldsymbol{c}\in C$, we have $H\boldsymbol{c}^{\top}=\mathbf 0$, which implies
$$
H_{j_1} \boldsymbol{c}_{j_1}^\top + H_{j_2} \boldsymbol{c}_{j_2}^\top + \cdots + H_{j_r} \boldsymbol{c}_{j_r}^\top = \mathbf 0.
$$
Hence, the union of columns of blocks $\mathcal{A}:=\{H_{j_1},\dots,H_{j_r}\}$ is an $\mathbb F_q$-linearly dependent set. Since $r\le t$, we may adjoin any additional $t-r$ blocks of $H$ to $\mathcal{A}$; the union of the columns of these $t$ blocks is still $\mathbb F_q$-linearly dependent, contradicting the hypothesis.

which contradicts the hypothesis.
\item Assume that $t+1$ block indices are $J=\{j_1,\dots,j_{t+1}\}$. Since the union of the columns of the block matrices $H_{j_1},\dots,H_{j_{t+1}}$ is $\mathbb F_q$-linearly dependent, there exist vectors $\boldsymbol{c}_{j_r}\in\mathbb F_q^{n_{j_r}}$, not all zero, such that
$$
H_{j_1}\boldsymbol{c}_{j_1}^{\top} + H_{j_2}\boldsymbol{c}_{j_2}^{\top} + \cdots + H_{j_{t+1}}\boldsymbol{c}_{j_{t+1}}^{\top} = 0.
$$
Define the non-zero vector $\boldsymbol{c}\in\mathbb F_q^n$ by placing these block vectors $\boldsymbol{c}_{j_r}$ in the corresponding block positions and zeros elsewhere. Then 
$$
H \boldsymbol{c}^T = \sum_{i=1}^s H_i \boldsymbol{c}_i^T = H_{j_1}\boldsymbol{c}_{j_1}^T + \cdots + H_{j_{t+1}}\boldsymbol{c}_{j_{t+1}}^T = 0.
$$
It follows that the non-zero vector $\boldsymbol{c}$ is in $C$. Since the number of non-zero blocks of $\boldsymbol{c}$ is at most $t+1$, we have $d_{\pi}(C)\le t+1$.
\end{enumerate}
\end{proof}

The following lemma from \cite{Feng2006} allows us to derive a $\pi$-distance analogue of a fundamental inequality from the Delsarte–Levenshtein theory for $\mathbb{F}_q$-linear error-block codes (see, e.g., \cite[Theorem 4.5]{Levenshtein}).

\begin{Lemma}\cite[Theorem 2.1]{Feng2006}\label{makh}
Let $C$ be an $\mathbb F_q$-linear error-block code of type $\pi = [n_1][n_2]\dots[n_s]$ with parameters $[n,k,d_\pi]_q$. Then the Singleton-type bound is
\begin{equation}
n - k \ge n_1 + n_2 + \dots + n_{d_\pi-1} \quad \left( \iff k \le n_{d_\pi} + n_{d_\pi+1} + \dots + n_s \right).
\end{equation}
\end{Lemma}
An $\mathbb F_q$-linear error-block code is called maximum distance separable (MDS) if equality holds in the Singleton-type bound.
\begin{Corollary}\label{sari}
Let $C$ be an $\mathbb{F}_q$-linear error-block code of type $\pi=[n_1][n_2]\dots[n_s]$ with parameters $[n,k,d_\pi]_q$, and let $d_\pi^\perp$ denote the $\pi$-distance of the dual code $C^\perp$. Then
$d_\pi+d_\pi^\perp\le n+2$.
\end{Corollary}
\begin{proof}
    Since $n_i \ge 1$ for all $i \in \{1, \dots, s\}$, applying Lemma~\ref{makh}  to $C$ and $C^\perp$ gives
$$
n - k  \ge d_{\pi} - 1, \quad k \ge d_{\pi}^\perp - 1,
$$
which proves the result.
\end{proof}
%=====================================================================
%Section: Mixed orthogonal array
%=====================================================================

\section{Mixed orthogonal array}
For a positive integer $q$, the set $H_q = \{0, 1, \ldots, q-1\}$ represents an alphabet of size $q$. 
\begin{Definition}
A mixed orthogonal array $\mathrm{MOA}(M, s, (q_1, \dots, q_s), t)$ is an $M \times s$ array whose $j$-th column takes its entries from an alphabet $H_{q_j}$, such that in any $M \times t$ subarray, each of the possible $t$-tuples occurs the same number of times.
\end{Definition}
In this paper, we restrict attention to mixed orthogonal arrays whose alphabet sizes are $H_{q_i}=\mathbb{F}_{q^{n_i}}$. For convenience, we also assume that $n_1\ge n_2\ge\cdots\ge n_s$, which can always be achieved by reordering the columns.

Each row of $\mathrm{MOA}(M, s, (q^{n_1},\dots, q^{n_s}), t)$ can  be viewed as an element of $\mathbb{F}_{q^{n_1}} \times \mathbb{F}_{q^{n_2}} \times \cdots \times \mathbb{F}_{q^{n_s}}$.
If we select $t$ columns indexed by $i_1, i_2, \ldots, i_t$, then the entries in these columns come from $\mathbb{F}_{q^{n_{i_1}}} \times \mathbb{F}_{q^{n_{i_2}}} \times \cdots \times \mathbb{F}_{q^{n_{i_t}}}$.  Consequently, the corresponding $M\times t$ subarray contains exactly  $q^{n_{i_1}+n_{i_2}+\cdots+n_{i_t}}$ distinct possible rows. By the definition of a mixed orthogonal array, each of these rows occurs the same number of times. Since the array has $M$ rows, each occurs exactly
\begin{align}\label{index}
\frac{M}{q^{n_{i_1}+n_{i_2}+\cdots+n_{i_t}}}
\end{align}
times. This number is called the index corresponding to the column subset $\{i_1,i_2,\ldots,i_t\}$ and is denoted by $\lambda_{i_1,i_2,\ldots,i_t}$. In general, these indices depend on the chosen column subset. Consequently, unlike ordinary orthogonal arrays, a mixed orthogonal array does not possess a single global index. However, when $n_1=n_2=\cdots=n_s$, all indices $\lambda_{i_1,i_2,\ldots,i_t}$ coincide, and the array has a single index, usually denoted by $\lambda$.

To save space, some examples present the array in its transposed form as an $s \times M$ matrix rather than an $M \times s$ matrix. In these cases, the term ``transposed'' is explicitly noted.

\begin{Example}\label{e1}
Consider the following array in which  the first  column takes values in $\mathbb F_4$ and the other columns take values in $\mathbb F_2$ (transposed):
$$
\begin{array}{cccccccc}
0 & 0 & 1 & 1 & \alpha & \alpha & \alpha+1 & \alpha+1 \\
0 & 1 & 0 & 1 & 0 & 1 & 0 & 1 \\
0 & 1 & 0 & 1 & 1 & 0 & 1 & 0 \\
0 & 1 & 1 & 0 & 0 & 1 & 1 & 0 \\
0 & 1 & 1 & 0 & 1 & 0 & 0 & 1
\end{array}.
$$
This array is an $\mathrm{MOA}(8,5,(2^2,2,2,2,2),2)$ with $n_1=2, n_i=1$ for all $i=2,3,4,5$. We have 
$$
\lambda_{i_1,i_2}=
\begin{cases}
\frac{8}{2^2}=2, & 2\le i_1<i_2\le5,\\[4pt]
\frac{8}{2^3}=1, & i_1=1,\; 2\le i_2\le5 .
\end{cases}
$$
\end{Example}

It is well known that any $t$ columns of an orthogonal array of strength $t$ over $\mathbb{F}_q$ are $\mathbb{F}_q$-linearly independent; see \cite[Theorem 3.27]{HSS2012}. The following theorem shows that the same property holds for mixed orthogonal arrays.

\begin{Theorem}\label{ma16}
Let $C$ be a mixed orthogonal array $\mathrm{MOA}(M,s,(q^{n_1},\dots,q^{n_s}),t)$. Then any $t$ columns of $C$ are $\mathbb{F}_q$-linearly independent.
\end{Theorem}
\begin{proof}
Let $u_1,u_2,\ldots,u_s$ denote the columns of $C$.  Suppose that
\begin{align}\label{run}
    \alpha_1 u_{j_1}+\alpha_2 u_{j_2}+\cdots+\alpha_t u_{j_t}=0, \quad\alpha_i\in\mathbb{F}_q.
\end{align}
Since $C$ is an $\mathrm{MOA}$ of strength $t$, every tuple in  $\mathbb{F}_{q^{n_{j_1}}}\times\cdots\times\mathbb{F}_{q^{n_{j_t}}}$ appears among the rows of the $M\times t$ subarray formed by these columns. In particular, all vectors $(0,\dots,0,1,0,\dots,0)$, in which the $1$ appears in the $i$-th position and $0$ elsewhere, occur among the rows of the subarray $u_{j_1},u_{j_2},\ldots,u_{j_t}$.
Substituting these rows into \eqref{run} yields $\alpha_1=\alpha_2=\cdots=\alpha_t=0$. Hence the columns $u_{j_1},u_{j_2},\ldots,u_{j_t}$ are $\mathbb{F}_q$-linearly independent.
\end{proof}

%=====================================================================

\subsection{MDS and almost-MDS mixed orthogonal arrays}

Throughout this paper, where mixed orthogonal arrays $C$ are considered with the Hamming metric, the minimum Hamming distance of $C$ is denoted by $d_H$ or, equivalently, by $d_H(C)$.

The following theorem extends the Singleton bound for orthogonal arrays (see \cite[Theorem 4.20]{HSS2012}) to mixed orthogonal arrays.

\begin{Theorem}[The Singleton-type bound]\label{Singleton}
Let $C$ be an $\mathrm{MOA}(M, s, (q^{n_1},\dots, q^{n_s}), t)$ with the minimum Hamming distance $d_H$. Then 
\begin{equation}\label{eq:m61}
    \prod_{i=1}^{t} q^{n_i} \;\leq\; M \;\leq\; \prod_{i=d_H}^{s} q^{n_i} \;\leq\; \prod_{i=1}^{\,s-d_H+1} q^{n_i}.
\end{equation}
\end{Theorem}

\begin{proof}
For any selection of $t$ columns $\{i_1,\dots,i_t\}$, the corresponding index $\lambda_{i_1,\ldots,i_t}= M/q^{n_{i_1} + \cdots + n_{i_t}}$ is an integer by the definition of MOAs. Thus $q^{n_{i_1}+\cdots+n_{i_t}}\mid M $, and hence $q^{n_{i_1}+\cdots+n_{i_t}}\leq M$. In particular, selecting the $t$ columns $\{1,2,\dots,t\}$ establishes the left-hand side of \eqref{eq:m61}. Notice that since $n_1\ge n_2\ge\cdots\ge n_s$, the lower bound $\prod_{i=1}^{t} q^{n_i}$ is the largest among all bounds obtained from other selections of $t$ columns.

To prove the upper bound, delete the first $d_H-1$ columns of $C$. Let $C'$ denote the resulting array with $s-d_H+1$ columns. We claim that the rows of $C'$ are pairwise distinct. Suppose, to the contrary, that two distinct rows of $C'$ become identical after deleting the first $d_H-1$ columns. Then these two rows agree in all coordinates from $d_H$ through $s$, and hence they can differ only among the first $d_H-1$ coordinates. It follows that their Hamming distance is at most $d_H-1$, contradicting the definition of the minimum distance $d_H$ (where $d_H$ is the minimum Hamming distance between any two distinct rows).
Therefore, the rows of $C'$ are pairwise distinct, so $M=|C'|$. Since the entries in column $i$ take values in an alphabet of size $q^{n_i}$, we have
\begin{equation}
    M=|C'| \leq \prod_{i=d_H}^{s} q^{n_i}.
\end{equation}
Finally, since $n_1\ge\cdots\ge n_s$, we have
$$
n_{d_H}\le n_1,\quad n_{d_H+1}\le n_2,\,\,\dots,\,\, n_s\le n_{s-d_H+1},
$$
which gives
$$
\prod_{i=d_H}^s q^{n_i}\le\prod_{i=1}^{s-d_H+1} q^{n_i}.
$$
\end{proof}
We define the upper bound $\prod_{i=d_H}^{s} q^{n_i}$ as the Singleton-type bound. A mixed orthogonal array $\mathrm{MOA}(M, s, (q^{n_1},\dots, q^{n_s}), t)$ is said to be maximal distance separable (MDS) if it attains this bound, that is,
$$M=\prod_{i=d_H}^{s} q^{n_i}.$$

\begin{Remark}
In the case of orthogonal arrays with equal block sizes $n_1 = \cdots = n_s$,  we have $ \prod_{i=d_H}^{s} q^{n_i} = \prod_{i=1}^{\,s-d_H+1} q^{n_i}$, which is precisely the classical Singleton-type bound for orthogonal arrays; see \cite[Theorem 4.20]{HSS2012}. However, this equality does not hold in general, as shown in Example \ref{ma50}.
\end{Remark}

\begin{Example}\label{ma50}
In the mixed orthogonal array $\mathrm{MOA}(8,5,(2^2,2,2,2,2),2)$ given in Example \ref{e1}, the minimum distance is $d_H = 3$. Thus
$$
\prod_{i=1}^{t} q^{n_i} =4\cdot 2 = 8= M=2\cdot 2\cdot 2=  \prod_{i=d_H}^{s} q^{n_i} < 16=4\cdot 2\cdot 2=\prod_{i=1}^{s-d_H+1} q^{n_i}.
$$
This example shows that the lower and upper bounds in \eqref{eq:m61} are satisfied, while the last inequality is strict. 
\end{Example}

Recall that for any selection of $t$ columns $\{i_1,\dots,i_t\}$ of the $\mathrm{MOA}$, the corresponding index is denoted by $\lambda_{i_1,\dots,i_t}$.
We define the minimum index by
$$
\lambda_{\min}:=\min_{\{i_1,\dots,i_t\}\subseteq\{1,\dots,s\}}
\lambda_{i_1,\dots,i_t}.
$$
Under the assumption $n_1 \ge n_2 \ge \cdots \ge n_s$, the quantity $q^{n_{i_1}+\cdots+n_{i_t}}$ is maximized when $\{i_1,\dots,i_t\}=\{1,2,\dots,t\}$. Therefore,
$$
\lambda_{\min}
= \min_{\{i_1,\dots,i_t\}}
\frac{M}{q^{n_{i_1}+\cdots+n_{i_t}}}
= \frac{M}{\max\limits_{\{i_1,\dots,i_t\}} q^{n_{i_1}+\cdots+n_{i_t}}}
= \frac{M}{\prod_{i=1}^{t} q^{n_i}}.
$$
In the case of OAs with equal block sizes $n_1=\cdots=n_s$, the minimum index $\lambda_{\min}$ coincides with the usual index $\lambda$. In this setting it is known that $d_H = s - t + 1$ if and only if the array is MDS if and only if $\lambda = 1$; see \cite[Theorems 4.20 and 4.21]{HSS2012}. For mixed orthogonal arrays, however, these equivalences need not hold in general. The following corollary, which follows from Theorem~\ref{Singleton}, identifies conditions under which the equivalence between the MDS property and $\lambda_{\min}=1$ remains valid.

\begin{Corollary}\label{ma60}
 Let $C$ be an $\mathrm{MOA}(M, s, (q^{n_1},\dots, q^{n_s}), t)$ with the minimum Hamming distance $d_H$.
 \begin{enumerate}
     \item If  $d_H = s - t + 1$, then $C$ is $\mathrm{MDS}$ if and only if $\lambda_{\min}=1$. 
    \item If $\sum_{j=d_H}^{s} n_j= \sum_{i=1}^{t} n_i $, then $C$ is $\mathrm{MDS}$ if and only if $\lambda_{\min}=1$.
  \end{enumerate}
\end{Corollary}
%In logic, $P \iff Q$ is true if both $P$ and $Q$ are always true.

\begin{Remark}\label{o-d}
\begin{enumerate}
    \item Example \ref{ma50} shows that the converse of the first statement in Corollary \ref{ma60} is not true in general. 
    \item In the case of OAs with equal block sizes $n_1=\cdots=n_s$, the condition $\sum_{j=d_H}^{s} n_j= \sum_{i=1}^{t} n_i $ is equivalent to $d_H = s - t + 1$.
    \item Example \ref{Lh} shows that without the condition $\sum_{j=d_H}^{s} n_j= \sum_{i=1}^{t} n_i $, the last statement in Corollary \ref{ma60} is not true in general. 
\end{enumerate}
\end{Remark}

\begin{Example}\label{Lh}
Consider the following array in which  the first  column takes values in $\mathbb F_4$ and the other columns take values in $\mathbb F_2$ (transposed):
$$
\begin{array}{cccccccc}
0 & 0 & 1 & 1 & \alpha & \alpha & \alpha+1 & \alpha+1 \\
0 & 1 & 0 & 1 & 0 & 1 & 0 & 1 \\
0 & 1 & 0 & 1 & 1 & 0 & 1 & 0
\end{array}.
$$

This array is an $\mathrm{MOA}(8,3,(2^2,2,2),2)$. One checks that $\lambda_{1,2}=1=\lambda_{1,3}$ and $\lambda_{2,3}=2$, and therefore $\lambda_{\min}=1$.
However, the minimum Hamming distance is $d_H=1$, and thus $M=8<16=2^{2}\cdot 2\cdot 2=\prod_{i=d_H}^{s} q^{n_i},$ which means $C$ is not MDS.
This array therefore shows that $\lambda_{\min}=1$ does not necessarily imply that $C$ is MDS. 
We now give an example illustrating that the converse of Corollary \ref{ma60} is also not true in general. Consider the following array in which the first column takes values in $\mathbb F_4$, and the other two columns take values in $\mathbb F_2$ (transposed):
$$
\begin{array}{cccccccccccccccc}
0 & 0 & 0 & 0 & 1 & 1 & 1 & 1 & \alpha & \alpha & \alpha & \alpha & \alpha+1 & \alpha+1 & \alpha+1 & \alpha+1 \\
0 & 0 & 1 & 1 & 0 & 0 & 1 & 1 & 0 & 0 & 1 & 1 & 0 & 0 & 1 & 1 \\
0 & 1 & 0 & 1 & 0 & 1 & 0 & 1 & 0 & 1 & 0 & 1 & 0 & 1 & 0 & 1
\end{array}.
$$
This array is an $\mathrm{MOA}(16,3,(2^2,2,2),2)$. Here $\lambda_{1,2}=\lambda_{1,3}=2$ and $\lambda_{2,3}=4$, which leads to $\lambda_{\min}=2$.
The minimum Hamming distance is $d_H=1$. Since $M=16=\prod_{i=d_H}^{s} q^{n_i}$, the array attains the Singleton-type bound and therefore $C$ is MDS.
\end{Example}

Let $C$ be an $\mathrm{MOA}(M,s,(q^{n_1},\dots,q^{n_s}),t)$ with minimum Hamming distance $d_H$. Define the Singleton defect of $C$ as
\begin{equation}
\delta(C) := \sum_{i=d_H}^s n_i - \log_q M.
\end{equation}
By \eqref{eq:m61}, we have $\delta(C)\ge 0$. Moreover, $C$ is MDS  if and only if $\delta(C) = 0$. We say that $C$ is almost-MDS if $\delta(C) = 1$, equivalently,
\begin{equation}
M = \frac{1}{q} \prod_{i=d_H}^{s} q^{n_i}.
\end{equation}

\begin{Theorem}\label{th:nearcon}
Let $C$ be a mixed orthogonal array $\mathrm{MOA}(M,s,(q^{n_1},\dots,q^{n_s}),t)$ with minimum Hamming distance $d_H$. If $ \sum_{i=1}^{t}n_i=\sum_{i=d_H}^{s}n_i-1,$ then $C$ is almost-MDS if and only if $\lambda_{\min}= 1$.
\end{Theorem}

\begin{proof}
Suppose that $C$ is almost-MDS. Then
$$
\lambda_{\min}=\frac{M}{\prod_{i=1}^{t} q^{n_i}}=\frac{\frac{1}{q} \prod_{i=d_H}^{s} q^{n_i}}{\prod_{i=1}^{t} q^{n_i}}=\frac{q^{\sum_{i=d_H}^{s}n_i-1}}{q^{\sum_{i=1}^{t}n_i}}=1.
$$
Conversely, suppose that $\lambda_{\min}=1$, i.e., $M=q^{\sum_{i=1}^{t}n_i}$. Using the hypothesis, $M = q^{(\sum_{i=d_H}^{s}n_i) - 1}$, which gives the result.
\end{proof}

\begin{Example}\label{ex:nearMDS}
Consider the mixed orthogonal array $\mathrm{MOA}\bigl(16,4,(2^2,2,2,2),3\bigr)$ (transposed):
$$
\begin{array}{cccccccccccccccc}
0 & \alpha & \alpha & 0 & \alpha & 0 & 0 & \alpha & 1 & \alpha+1 & \alpha+1 & 1 & \alpha+1 & 1 & 1 & \alpha+1 \\
0 & 0 & 0 & 0 & 1 & 1 & 1 & 1 & 0 & 0 & 0 & 0 & 1 & 1 & 1 & 1 \\
0 & 0 & 1 & 1 & 0 & 0 & 1 & 1 & 0 & 0 & 1 & 1 & 0 & 0 & 1 & 1 \\
0 & 1 & 0 & 1 & 0 & 1 & 0 & 1 & 0 & 1 & 0 & 1 & 0 & 1 & 0 & 1
\end{array}.
$$
The minimum Hamming distance of $C$ is $d_H=1$. In addition, $\sum_{i=1}^t n_i =4$ and $\sum_{i=d_H}^{s} n_i -1=4$. Furthermore,
$$
\lambda_{\min}=\frac{M}{\prod_{i=1}^{t}2^{n_i}}=\frac{16}{2^{4}}=1,
\qquad
\delta(C)=\sum_{i=1}^4 n_i-\log_2 16=1.
$$
Hence $C$ is almost-MDS with $\lambda_{\min}=1$.
\end{Example}
%=====================================================================
%Section: The duality of $\mathbb F_q$-linear MOAs
%=====================================================================

\section{Duality of $\mathbb F_q$-linear mixed orthogonal arrays}\label{S:1-3}
In this section, we define the trace duality for $\mathbb F_q$-linear MOAs, and we will show how this duality is related to the Euclidean dual of $\mathbb F_q$-linear error-block codes. 

Any extension field $\mathbb{F}_{q^{m}}$ has an $\mathbb F_q$-basis $\mathcal{B} = \{\beta_1, \beta_2,\ldots, \beta_m\}$ such that every element $x\in \mathbb F_{q^{m}}$ can be written as $x = \sum_{i=1}^{m} x_i \beta_i$ with $x_i\in \mathbb F_q$. The Gram matrix of $\mathcal{B}$ is an $m\times m$ matrix $M$ with entries  $M_{ij}:=\operatorname{Tr}(\beta_i\beta_j)$, where $ \operatorname{Tr} $ is the trace map  from $\mathbb{F}_{q^{m}}$ to $\mathbb{F}_q $ defined by $\operatorname{Tr}(x) =  x+ x^q + \cdots + x^{q^{m-1}}$ for any $x\in \mathbb F_{q^m}$. In particular, $\operatorname{Tr}(x)=x$ for all $x\in\mathbb F_q$. We define an $\mathbb F_q$-linear map
\begin{align}\label{rho-map}
 \rho : \mathbb{F}_{q^m} &\longrightarrow \mathbb F_q^{m}  \\ 
 x = \sum_{i=1}^{m} x_i \beta_i&\longmapsto (x_1, x_2, \dots, x_m)\nonumber.
\end{align}

\begin{Lemma}
Let $x, y\in \mathbb F_{q^m}$ and $M$ be the Gram matrix of $\mathcal{B}$. Then $\operatorname{Tr}(xy)=\rho(x) M \rho(y)^{T}$, where $T$ denotes the transpose.
\end{Lemma}
\begin{proof}
Let $x=\sum_{i=1}^m x_i\beta_i$ and $ y=\sum_{j=1}^m y_j\beta_j$. Thus $xy=\sum_{i=1}^m\sum_{j=1}^m x_iy_j\beta_i\beta_j$, and hence
\begin{align*}
\operatorname{Tr}(xy)
&=\sum_{i=1}^{m}\sum_{j=1}^{m} x_i y_j\,\operatorname{Tr}(\beta_i\beta_j)\\
&=\sum_{i=1}^{m}\sum_{j=1}^{m} x_i M_{ij} y_j=\begin{pmatrix} x_1 & \cdots & x_m \end{pmatrix}
M
\begin{pmatrix} y_1 \\ \vdots \\ y_m \end{pmatrix}
=\rho(x)\,M\,\rho(y)^T .
\end{align*}
\end{proof}

The $\mathbb F_q$-basis $\mathcal{B} = \{\beta_1, \beta_2,\ldots, \beta_m\}$ is called self-dual if  $\operatorname{Tr}(\beta_i\beta_j)=\delta_{i,j}$, where $\delta_{i,j}$ is  the Kronecker delta. For example, The basis $\{\alpha, \alpha^2\}$ of the field extension $ \mathbb F_4 = \mathbb F_2[\alpha]/\langle 1 + \alpha + \alpha^2 \rangle $ is self-dual because $\operatorname{Tr}(\alpha \alpha)=\operatorname{Tr}(\alpha^2 \alpha^2)=1$ and $\operatorname{Tr}(\alpha \alpha^2)=\operatorname{Tr}(\alpha^2 \alpha)=0$. If $\mathbb{F}_{q^m}$ has a self-dual basis, then the Gram matrix of the basis is the identity matrix.  It is well known that $\mathbb{F}_{q^m}$ admits a self-dual basis over $\mathbb{F}_q$ if and only if either $q$ is even or both $q$ and $m$ are odd; see \cite{Seroussi1980}.

\begin{Corollary}
Suppose that either $q$ is even or both $m$ and $q$ are odd. Then for any $x, y\in \mathbb F_{q^m}$,
$$\operatorname{Tr}(xy)=\rho(x) \rho(y)^{T}.$$
\end{Corollary}

For each $1\leq i\leq s$, let $\mathcal{B}_i$ be an $\mathbb F_q$-basis of $\mathbb F_{q^{n_i}}$,  and $\rho_i : \mathbb{F}_{q^{n_i}} \to\mathbb F_q^{n_i}$ be the $\mathbb F_q$-linear map associated with $\mathcal{B}_i$ defined in  \eqref{rho-map}. We extend the $\mathbb F_q$-linear map $\rho$ component-wise to the product space. 
By abuse of notation, we denote the extended map again by $\rho$:

\begin{align}\label{ma12}
    \rho\colon \mathbb{F}_{q^{n_1}}\times\mathbb{F}_{q^{n_2}}\times\cdots\times\mathbb{F}_{q^{n_s}} 
    &\longrightarrow
    \mathbb{F}_q^{n_1}\times\mathbb{F}_q^{n_2}\times\cdots\times\mathbb{F}_q^{n_s}=\mathbb{F}_q^n\\
    \mathbf{a}=(a_1, a_2, \ldots, a_s)
    &\longmapsto \Big(\rho_1(a_1), \rho_2(a_2), \ldots, \rho_s(a_s)\Big).\nonumber
\end{align}
We note that $\rho$ maps the ambient space of mixed orthogonal arrays to the ambient space of $\mathbb F_q$-linear error-block codes.

We define the trace inner product on $\mathbb{F}_{q^{n_1}}\times\cdots\times\mathbb{F}_{q^{n_s}}$ by
    $$
    \langle \mathbf{a},\,\mathbf{b}\rangle_{\operatorname{Tr}} = \sum_{i=1}^s \operatorname{Tr}_i(a_ib_i),
    $$
where $\mathbf{a} = (a_1, \dots, a_s)$ and $\mathbf{b} = (b_1, \dots, b_s)$ are in $\mathbb{F}_{q^{n_1}}\times\cdots\times\mathbb{F}_{q^{n_s}}$, and $ \operatorname{Tr}_{i} : \mathbb{F}_{q^{n_i}} \to \mathbb{F}_q $ denotes the trace map $\operatorname{Tr}_{i}(a) = a + a^q + \cdots + a^{q^{n_i-1}}.$
For a mixed orthogonal array $C$, we define the trace dual as
$$
C^{\perp_{\text{Tr}}} = \Big\{ \mathbf{b} \in \mathbb{F}_{q^{n_1}}\times\cdots\times\mathbb{F}_{q^{n_s}} \colon \langle \mathbf{c},\, \mathbf{b}\rangle_{\operatorname{Tr}}= 0,\, \text{ for all }\, \mathbf{c}  \in C \Big\}.
$$
The next lemma is immediate.
\begin{Lemma}\label{ma35} With the above notation, the following statements hold.
    \begin{enumerate}
     \item The map $\rho$ is an $\mathbb F_q$-linear isomorphism.
      \item Let $M_i$ be the Gram matrix of the $\mathbb{F}_q$-basis $\mathcal{B}_i$. Then
    $$\langle \mathbf{a},\,\mathbf{b}\rangle_{\operatorname{Tr}} = \sum_{i=1}^s \rho_i(a_i) M_i\rho_i(b_i)^T.$$
    \item Let $M=\operatorname{diag}(M_1,M_2,\ldots,M_s)$ be the block diagonal matrix. Then
    $$\langle \mathbf{a},\,\mathbf{b}\rangle_{\operatorname{Tr}} = \rho(\mathbf{a})M\rho(\mathbf{b})^T.$$  
\end{enumerate}
\end{Lemma}

A mixed orthogonal array is called $\mathbb F_q$-linear if its set of rows forms an $\mathbb F_q$-linear subspace of $\mathbb{F}_{q^{n_1}}\times\mathbb{F}_{q^{n_2}}\times\cdots\times\mathbb{F}_{q^{n_s}}.$ The following theorem shows that the image of an $\mathbb F_q$-linear MOA under $\rho$ is an $\mathbb F_q$-linear error-block code and that $\rho$ preserves both distance and duality between $\mathbb F_q$-linear MOAs and $\mathbb F_q$-linear error-block codes. These results play an important role in the subsequent sections.

\begin{Theorem}\label{ma15} Let $C$ be an $\mathbb{F}_q$-linear $\mathrm{MOA}$. Then the following statements hold.
\begin{enumerate}
    \item $\rho(C)$ is an $\mathbb F_q$-linear error-block code.
    \item If $d_H$ denotes the Hamming distance and $d_\pi$ the minimum $\pi$-distance, then
    $$d_H(C)=d_\pi\big(\rho(C)\big).$$
    \item If either $q$ is even or $q$ is odd and $n_i$ is odd for all $i$, then $\rho$ preserves duality, that is
    $$\rho\big(C^{\perp_{\operatorname{Tr}}}\big)=\big(\rho(C)\big)^{\perp}.$$
    \item If $n_i=1$ for all $i\in\{1, 2, \ldots, s\}$, then $C^{\perp_{\text{Tr}}} $ is the Euclidean dual.
\end{enumerate}  
\end{Theorem}

\begin{proof}
The statement 1 follows immediately from Lemma~\ref{ma35}. To prove the second statement, let $\mathbf a=(a_1,\dots,a_s)\in
\mathbb F_{q^{n_1}}\times\cdots\times\mathbb F_{q^{n_s}}$. For any $1\leq i\leq s$, we know $\rho_i$ is an $\mathbb F_q$-isomorphism. Thus $\rho_i(a_i)=0$ if and only if $a_i=0$. 
The $\pi$-weight $w_\pi(\rho(\mathbf a))$ equals the number of indices $j$ for which the $j$-th block $\rho_j(a_j)$ is non-zero. Hence,
$$
w_H(\mathbf a)=\#\{j:a_j\neq0\}=\#\{j:\rho_j(a_j)\neq0\}=w_\pi(\rho(\mathbf a)).
$$  
For the statement 3, we note that if either $q$ is even or $q$ is odd and $n_i$ is odd
for all $i$, then each field $\mathbb{F}_{q^{n_i}}$ admits a self-dual
$\mathbb{F}_q$-basis $\mathcal B_i$. The corresponding Gram matrices of $\mathcal B_i$ satisfy $M_i=I$. Hence Lemma \ref{ma35} gives $\langle \mathbf a,\mathbf b\rangle_{\operatorname{Tr}}=\rho(\mathbf a)\rho(\mathbf b)^T,$ which proves the result.
To prove the last item, we notice that the trace maps $\operatorname{Tr}_i:\mathbb F_{q^{n_i}}\to\mathbb F_q$ act as the identity map over $\mathbb F_q$. Since $n_i=1$, the maps $\operatorname{Tr}_i$ are the identity maps. So $\langle\, ,\,\rangle_{\operatorname{Tr}}$ is the usual inner product, which implies $C^{\perp_{\text{Tr}}} $ is the Euclidean dual.
\end{proof}

%=====================================================================
%Section: $\mathbb F_q$-linear MOAs
%=====================================================================

\section{$\mathbb F_q$-linear mixed orthogonal arrays}

This section establishes a structural correspondence between $\mathbb{F}_q$-linear MOAs and $\mathbb{F}_q$-linear error-block codes, providing an analogue to the classical duality between orthogonal arrays and linear codes (\cite[Theorem 4.6]{HSS2012}). The main results of this section are Theorems~\ref{necessary-condition} and~\ref{th:so}. Specifically, Theorem~\ref{necessary-condition} asserts that the image of an $\mathbb{F}_q$-linear MOA under the map $\rho$ is an $\mathbb{F}_q$-linear error-block code, whose dual code has $\pi$-distance bounded below by $t+1$. Conversely, Theorem~\ref{th:so} demonstrates that the inverse image $\rho^{-1}(C)$ of an $\mathbb{F}_q$-linear error-block code $C$ is an $\mathbb{F}_q$-linear MOA with strength $d_\pi(C^\perp) - 1$.

To prove these main results, we first establish Theorem~\ref{ma20}, which generalizes the independence-strength correspondence of orthogonal arrays (\cite[Theorem 3.29]{HSS2012}) to the mixed setting. As the converse to Theorem~\ref{ma16}, Theorem~\ref{ma20} asserts that if every set of $t$ columns of an $\mathbb F_q$-linear array is $\mathbb F_q$-linearly independent, then the array is an $\mathbb F_q$-linear MOA of strength $t$. The proof of this theorem relies on the following two auxiliary lemmas.

\begin{Lemma}\label{surjective}
Let $C$ be an $\mathbb F_q$-linear error-block code in $ \mathbb{F}_q^n =\mathbb{F}_q^{n_1}\times \mathbb{F}_q^{n_2} \times \ldots\times \mathbb{F}_q^{n_s}$ with the Euclidean dual $C^\perp$. Let $\{i_1,\ldots,i_t\}\subseteq\{1,\ldots,s\}$ be a fixed selection of $t$ blocks, and consider the $\mathbb F_q$-linear projection map
        \begin{align*}
         \mu : C &\to \mathbb{F}_{q}^{n_{i_1}} \times \cdots \times \mathbb{F}_{q}^{n_{i_t}}\\  
          \mathbf c=(\mathbf c_1,\ldots,\mathbf c_s)&\mapsto(\mathbf c_{i_1},\ldots,\mathbf c_{i_t}).
        \end{align*}
If $\mu$ is not surjective, then $d_\pi(C^\perp)\le t$.
\end{Lemma}
\begin{proof}
Let $ G $ be a generator matrix of $C$ as
$$
G = \Big[G_1 \mid G_2 \mid \cdots \mid G_s\Big],
$$
where $G_j$ has $n_j$ columns. Therefore, the generator matrix of $\mu(C)$ is 
$$
G':= \Big[G_{i_1} \mid \cdots \mid G_{i_t}\Big].
$$
We know that the rank of $G'$ is the dimension of the vector space spanned by its rows. Thus, if we assume that $\mu$ is not surjective, then 
$$\operatorname{rank}(G')=\dim(\mu(C)) < n_{i_1} + \cdots + n_{i_t}.$$
Thus by rank–nullity theorem, $\dim(\operatorname{null}(G')) \ge 1$. Hence, there exists at least one non-zero vector
$$
\mathbf h'=(\mathbf h_{i_1},\ldots,\mathbf h_{i_t})\in\mathbb F_q^{n_{i_1}}\times\cdots\times\mathbb F_q^{n_{i_t}}
$$
in the null space of $G'$, i.e., $G'(\mathbf{h}')^T=0$, where $T$ denotes the transpose. Let $g_1,\ldots,g_k$ denote the rows of $G'$. Then $g_1\cdot\mathbf h'=\cdots=g_k\cdot\mathbf h'=0,$ where $\cdot$ denotes the Euclidean inner product. Since the rows of $G'$ generate $\mu(C)$, every
$\mathbf c'=(\mathbf c_{i_1},\ldots,\mathbf c_{i_t})\in\mu(C)$ can be written as
$$
\mathbf c'=a_1g_1+\cdots+a_kg_k
$$
for some $a_1,\ldots,a_k\in\mathbb F_q$, and hence $\mathbf c'\cdot\mathbf h'=0.$
Now define $\mathbf h\in\mathbb F_q^n$ by inserting the blocks
$\mathbf h_{i_1},\ldots,\mathbf h_{i_t}$ into the corresponding positions
and zeros elsewhere. Then for any $\mathbf c=(\mathbf c_1,\ldots,\mathbf c_s)\in C$,
$$
\mathbf c\cdot\mathbf h=\sum_{k=1}^t \mathbf c_{i_k}\cdot\mathbf h_{i_k}=\mathbf c'\cdot\mathbf h'=0.
$$
Thus $\mathbf h\in C^\perp$. Since $\mathbf h$ has at most $t$ nonzero blocks, its $\pi$-weight is at most $t$, and therefore $d_\pi(C^\perp)\le t.$
\end{proof}

\begin{Lemma} \label{coset}
Let $A$ and $B$ be $\mathbb F_q$-vector spaces, and let $f: A \to B$ be a surjective $\mathbb F_q$-linear map. Then for every $b \in B$,
$$|f^{-1}(b)| = |\ker f|.$$
\end{Lemma}

\begin{proof}
Since $f$ is an $\mathbb F_q$-linear map, $\ker  f$ is a subspace of $A$. Since $f$ is surjective, for every $b \in B$, there exists $a_0 \in A$ such that $f(a_0) = b$. We know that 
$$f^{-1}(b) = a_0 + \ker f$$
is a coset of the subspace $\text{ker} f$.
Define
\begin{align*}
 \phi: \ker f &\to a_0 +\ker f  \\
k &\mapsto a_0 + k.
\end{align*}
This map is a bijection. Hence $|\ker f|=|a_0 +\ker f|=|f^{-1}(b)|.$
\end{proof}

The following theorem establishes the converse of Theorem~\ref{ma16} for $\mathbb{F}_q$-linear arrays, thereby extending the classical characterization of orthogonal arrays (\cite[Theorem 3.29]{HSS2012}) to the mixed setting.

\begin{Theorem}\label{ma20}
Let $C$ be an $\mathbb F_q$-linear $M \times s$ array whose rows lie in $\mathbb{F}_{q^{n_1}}\times\mathbb{F}_{q^{n_2}}\times\cdots\times\mathbb{F}_{q^{n_s}}$. If every set of $t$ columns of $C$ is $\mathbb F_q$-linearly independent, then $C$ is an $\mathbb F_q$-linear mixed orthogonal array $\mathrm{MOA}\left(M, s, (q^{n_1}, \dots, q^{n_s}), t\right)$.
\end{Theorem}
\begin{proof}
Since $\rho$ is an $\mathbb F_q$-isomorphism, $t$ columns $V_1, V_2, \ldots, V_t$ of $C$ are $\mathbb F_q$-linearly independent if and only if the union of the columns of the blocks $\rho(V_1), \rho(V_2), \ldots, \rho(V_t)$ is an $\mathbb F_q$-linearly independent set. Therefore, by Lemma~\ref{ma14}, we obtain $d_\pi((\rho(C))^\perp)\geq t+1$.\\
For a fixed selection of $t$ columns $\{i_1,\ldots,i_t\}$ of $C$, we define the projection maps $\mu'$ and $\mu$ as shown in the following diagram:
  \begin{align}
\begin{array}{ccc}
C & \xrightarrow{\quad \mu'\quad} & \mathbb{F}_{q^{n_{i_1}}}\times\mathbb{F}_{q^{n_{i_2}}}\times\cdots\times\mathbb{F}_{q^{n_{i_t}}} \\
\Big\downarrow{\scriptstyle \rho} &  & \Big\downarrow{\scriptstyle \rho} \\
\rho(C) & \xrightarrow{\quad \mu \quad} & \mathbb{F}_{q}^{n_{i_1}}\times\mathbb{F}_{q}^{n_{i_2}}\times\cdots\times\mathbb{F}_{q}^{n_{i_t}}.
\end{array}
\end{align} 
Since $d_\pi((\rho(C))^\perp)\ge t+1$, Lemma~\ref{surjective} implies that $\mu$ is surjective. Since $\rho$ is an isomorphism, it follows that $\mu'$ is also surjective, and hence $ |\mu'(C)|=q^{n_{i_1} + \cdots + n_{i_t}}.$ By the first isomorphism theorem,
$$ q^{n_{i_1} + \cdots + n_{i_t}}=|\mu'(C)|= \frac{|C|}{|\ker(\mu')|},$$
and hence
\begin{align}\label{First11}
    |\ker(\mu')| = \frac{M}{q^{n_{i_1} + \cdots + n_{i_t}}}
\end{align}
Using Lemma \ref{coset} and equation \eqref{First11}, for any  $\boldsymbol{c}' \in  \mathbb{F}_{q^{n_{i_1}}} \times \mathbb{F}_{q^{n_{i_2}}}\times \cdots \times \mathbb{F}_{q^{n_{i_t}}}$,  we have 
$$
|(\mu')^{-1}(\boldsymbol{c}')|=\frac{M}{q^{n_{i_1} + \cdots + n_{i_t}}}.
$$
Hence there are $\frac{M}{q^{n_{i_1} + \cdots + n_{i_t}}}$ codewords in $C$ that map to the  $t$-tuple $\mathbf{c}'$ under the map $\mu'$. \\
Therefore, for every selection $\{i_1,\ldots,i_t\}$ of $t$ columns of $C$, each $t$-tuple appears exactly
$$\frac{M}{q^{n_{i_1} + n_{i_2} + \cdots + n_{i_t}}}.$$  Therefore, $C$ is an $ \mathrm{MOA}\left(M, s, (q^{n_1},\dots, q^{n_s}), t\right)$.
\end{proof}

The next two theorems present the main results of this section, which establish a fundamental correspondence between $\mathbb{F}_q$-linear MOAs and $\mathbb{F}_q$-linear error-block codes.

\begin{Theorem}\label{necessary-condition}
 Let $C$ be an $\mathbb F_q$-linear mixed orthogonal array $\mathrm{MOA}\left(M, s, (q^{n_1}, \dots, q^{n_s}), t\right)$. Then the following statements hold.
  \begin{enumerate}
     \item $C$ is $\mathbb F_q$-isomorphic to the linear error-block code $\rho(C)$ of length $n=
     \sum_{i=1}^s n_i$ and of type $\pi=[n_1][n_2]\ldots[n_s]$.
     \item $d_\pi((\rho(C))^\perp)\geq t+1$.
     \item If the $\mathrm{MOA}$ has strength $t$ but not $t+1$, then $d_\pi((\rho(C))^\perp)= t+1$.
 \end{enumerate}
\end{Theorem}

\begin{proof}
The first statement is an immediate consequence of Theorem \ref{ma15}. To prove the second one, we note that by Theorem \ref{ma16}, any $t$ columns $V_1, V_2, \ldots, V_t$ of $C$ are $\mathbb F_q$-linearly independent. 
Since $\rho$ is an $\mathbb F_q$-isomorphism, the union of columns of the corresponding $t$ blocks $\rho(V_1), \rho(V_2), \ldots, \rho(V_t)$ in $\rho(C)$ is an $\mathbb F_q$-linearly independent set. Hence, by Lemma~\ref{ma14},
$$d_{\pi}\Big(\big(\rho(C)\big)^{\perp}\Big)\ge t+1.$$
To prove the last statement,  assume that every set of $t+1$ columns of $C$ is linearly independent. Then by Theorem \ref{ma20}, $C$ is a mixed orthogonal array $\mathrm{MOA}(M, s, (q^{n_1}, \dots, q^{n_s}), t+1)$, which contradicts the hypothesis because $C$ has strength $t$ but not $t+1$. Therefore, there exists a set of $t+1$ columns of $C$ that is linearly dependent. Since $\rho$ is an $\mathbb F_q$-isomorphism, there exists a set of $t+1$ blocks of $\rho(C)$ whose union of columns is $\mathbb F_q$-linearly dependent.  Applying Lemma \ref{ma14},  $d_\pi((\rho(C))^\perp)\leq t+1,$ which implies $d_\pi((\rho(C))^\perp)= t+1$ together with the second statement.
\end{proof}

\begin{Theorem}\label{th:so}
Let $C \subseteq\mathbb F_q^n$ be an $\mathbb{F}_q$-linear error-block code of length $n=\sum_{i=1}^s n_i$, type $\pi=[n_1][n_2]\ldots[n_s]$, and dimension $k$. Then  $\rho^{-1}(C)$ is an $\mathbb F_q$-linear mixed orthogonal array 
     $$\mathrm{MOA}\left(q^k, s, (q^{n_1}, \dots, q^{n_s}), d_\pi(C^{\perp})-1\right).$$
\end{Theorem}

\begin{proof}
By Theorem \ref{ma15}, the map $\rho$ is an $\mathbb{F}_q$-isomorphism, and hence $\rho^{-1}(C)$ is an $\mathbb{F}_q$-linear array in $\mathbb{F}_{q^{n_1}}\times\cdots\times\mathbb{F}_{q^{n_s}}$ with $M:=|C|=q^k$ rows. Let $t := d_\pi(C^\perp)-1.$ Suppose that there exists a set of $t$ columns of $\rho^{-1}(C)$ that is $\mathbb F_q$-linearly dependent. Applying $\rho$, we see that the union of the corresponding $t$ blocks of $C$ is $\mathbb{F}_q$-linearly dependent. By Lemma \ref{ma14}, this implies $d_\pi(C^\perp)\le t$, which contradicts the definition of $t$. Hence, any $t$ columns of $\rho^{-1}(C)$ are $\mathbb{F}_q$-linearly independent. By Theorem \ref{ma20}, it follows that $\rho^{-1}(C)$ is an $\mathbb{F}_q$-linear mixed orthogonal array of strength $t$. That is, $\rho^{-1}(C)$ is an $\mathrm{MOA}\left(M, s, (q^{n_1},\dots,q^{n_s}), d_\pi(C^\perp)-1\right)$.
\end{proof}

%=====================================================================
%Section: Irredundant MOAs
%=====================================================================

\section{Irredundant mixed orthogonal arrays}

In this section we study irredundant mixed orthogonal arrays (IrMOAs).  The section contains two main results. The first concerns the extremal case $t=\lfloor s/2\rfloor$ and shows that, under a  condition on the block sizes, an IrMOA is MDS if and only if its minimum index is equal to $1$ (Theorem\ref{cor:extreme_cases}). The second main result gives a method for constructing $\mathbb{F}_q$-linear IrMOAs from $\mathbb{F}_q$-linear error-block codes by characterizing when the preimages of a code and its dual under $\rho$ are irredundant mixed orthogonal arrays (Theorem \ref{latest}).

\begin{Definition}\cite{GBZ2016}
A mixed orthogonal array $\mathrm{MOA}\left(M, s, (q^{n_1},\dots, q^{n_s}), t\right)$ is called irredundant if the rows of any $M\times (s-t)$ subarray are pairwise distinct. Such an array is denoted by $\mathrm{IrMOA}(M, s, (q^{n_1},\dots, q^{n_s}), t)$.
\end{Definition}

The following lemma characterizes irredundancy in terms of the minimum Hamming distance.

\begin{Lemma}\cite[Lemma 3]{SSCZ2022}\label{m88}
A mixed orthogonal array $\mathrm{MOA}(M,s,(q^{n_1},\dots,q^{n_s}),t)$ with minimum Hamming distance $d_H$ is irredundant if and only if $d_H\ge t+1$. 
\end{Lemma}

%\begin{proof}
%If $d_H\ge t+1$ then no two distinct rows of $\mathrm{MOA}$ can coincide on $s-t$ or more coordinates, so every $M \times (s-t)$ subarray has pairwise distinct rows, i.e., it is irredundant. Conversely, if $d_H \le t$, then two rows agree on at least $s-t$ coordinates, contradicting irredundancy.    
%\end{proof}

\begin{Example}  
The array $C$ in Example~\ref{e1} is an $\mathrm{IrMOA}$ since $d_H(C)=3=t+1$.
\end{Example}

The following theorem, which is an immediate consequence of Theorem~\ref{Singleton} and Lemma~\ref{m88}, gives a necessary condition for the existence of IrMOAs.

\begin{Theorem}\label{m66}
Let $C$ be an $\mathrm{IrMOA}(M, s, (q^{n_1},\dots, q^{n_s}), t)$. Then
$$\prod_{i=1}^{t} q^{n_i}\le \prod_{i=t+1}^{s} q^{n_i}.$$
\end{Theorem}

\begin{proof}
Applying Theorem \ref{Singleton}, Lemma \ref{m88}, and the assumption $n_1\geq n_2\geq \ldots \geq n_s$, we obtain
$$q^{n_1}q^{n_2}\ldots, q^{n_t}\leq q^{n_{d_H}}q^{n_{d_H+1}}\ldots q^{n_s}\leq q^{n_{t+1}}q^{n_{t+2}}\ldots q^{n_s}.$$
\end{proof}

The next corollary says that the strength $t$ of any IrMOA is bounded by the number of columns.
    
\begin{Corollary} \label{m87}
Let $C$ be an $\mathrm{IrMOA}(M, s, (q^{n_1},\dots, q^{n_s}), t)$. Then  $2t\leq s$.
\end{Corollary}

\begin{proof}
By the assumption $n_1\geq n_2\geq \ldots \geq n_s$,  we have $q^{n_i}\ge q^{n_{t+1}}$ for every $i\le t$,  and $q^{n_j}\le q^{n_{t+1}}$ for every $j\ge t+1$.  Applying Corollary \ref{m66}, we obtain
$$
(q^{n_{t+1}})^t\leq q^{n_1}q^{n_2}\cdots q^{n_t} \leq q^{n_{t+1}}q^{n_{t+2}}\cdots q^{n_s} \;\le\; (q^{n_{t+1}})^{s-t},
$$
which gives $t\le s-t$, i.e., $2t\leq s$. 
\end{proof}

%=====================================================================

\subsection{MDS irredundant mixed orthogonal arrays}

By Corollary~\ref{m87}, any irredundant mixed orthogonal array must satisfy $2t\le s$.  Since $t$ is an integer, this implies $t \le \left\lfloor \frac{s}{2}\right\rfloor .$ When $t=\lfloor s/2\rfloor$, irredundant orthogonal arrays $\mathrm{IrOA}(M,s,q^{n},t)$ with $M=q^t$ (equivalently, with index $\lambda=1$) are known to be equivalent to classical MDS codes \cite{B2019,GALRZ2015}.  The following theorem and corollary generalize this result to irredundant mixed orthogonal arrays.

\begin{Theorem}\label{cor:extreme_cases}
Let $C$ be an $\mathrm{IrMOA}(M,s,(q^{n_1},\dots,q^{n_s}),t)$ satisfying $\sum_{j=t+1}^{2t+1} n_j = \sum_{i=1}^{t} n_i$, and assume that $t=\lfloor s/2\rfloor$. Then $C$ is $\mathrm{MDS}$ if and only if $\lambda_{\min}=1.$
\end{Theorem}

\begin{proof}
Observe that $t=\lfloor s/2\rfloor$ holds if and only if $s\in\{2t,\,2t+1\}$. We consider these two cases separately.

\textbf{Case 1: $s=2t$.}
Since $C$ is irredundant, the number of rows of $ M $ cannot exceed the number of distinct rows in any $M \times (s-t)$-subarray, that is $M\leq \prod_{i=1}^{s-t} q^{n_i}$ by choosing the particular $s-t$ columns $\{1, 2, \ldots, s-t\}$. On the other side, by Lemma \ref{m88}, $d_H\geq t+1$. Thus $s-t\geq s-d_H+1$.  Applying Theorem \ref{Singleton}, we have
\begin{align}\label{eq:m611}
    \prod_{i=1}^{t} q^{n_i} \;\leq\; M \;\leq\; \prod_{i=d_H}^{s} q^{n_i} \;\leq\; \prod_{i=1}^{\,s-d_H+1} q^{n_i}\leq \prod_{i=1}^{\,s-t} q^{n_i}.
\end{align}
Since $s=2t$, the lower and upper bounds in \eqref{eq:m611} coincide. Hence $d_H=s-t+1=t+1.$ Consequently, by Corollary~\ref{ma60}, the array $C$ is MDS if and only if $\lambda_{\min}=1$.

\textbf{Case 2: $s=2t+1$.}
Since $C$ is irredundant, Lemma~\ref{m88} gives $d_H\ge t+1$. 
Suppose that $d_H\ge t+3$. Then
$$s-d_H+1 \le (2t+1)-(t+3)+1 = t-1.$$
Using the Singleton-type bound \eqref{eq:m61}, we obtain
$$
\prod_{i=1}^{t} q^{n_i}\le M\le \prod_{i=1}^{s-d_H+1} q^{n_i}\le \prod_{i=1}^{t-1} q^{n_i},
$$
which is impossible since $n_t\ge 1$. Therefore $d_H\le t+2$, and hence $d_H\in\{t+1,t+2\}.$ If $d_H=t+2$, then $t=s-d_H+1$, and Corollary~\ref{ma60} implies that $C$ is MDS if and only if $\lambda_{\min}=1$.
If $d_H=t+1$, then by definition of $\lambda_{\min}$ we have  $M = \lambda_{\min} \prod_{i=1}^t q^{n_i}$. Substituting this into \eqref{eq:m61} and using the hypothesis $\sum_{j=t+1}^{2t+1} n_j = \sum_{i=1}^t n_i$, we obtain
$$ \lambda_{\min} \prod_{i=1}^t q^{n_i} \le \prod_{j=t+1}^{2t+1} q^{n_j} = \prod_{i=1}^t q^{n_i}. $$
That is $\lambda_{\min} \le 1$, which leads to $\lambda_{\min} = 1$. Now applying the last statement of Corollary \ref{ma60} implies that $C$ is $\mathrm{MDS}$.
\end{proof}

Since $\rho$ is an $\mathbb F_q$-linear isomorphism that preserves distance (Theorem~\ref{ma15}), the $\mathbb F_q$-mixed orthogonal array $C$ is MDS if and only if the $\mathbb F_q$-linear error-block code $\rho(C)$ is MDS. Thus, we have the following corollary.  

\begin{Corollary}
Let $C$ be an $\mathbb F_q$-linear $\mathrm{IrMOA}(M,s,(q^{n_1},\dots,q^{n_s}),t)$ satisfying $\sum_{j=t+1}^{2t+1} n_j = \sum_{i=1}^{t} n_i$ and assume that $t=\lfloor s/2\rfloor$. Then $\rho(C)$ is an MDS $\mathbb F_q$-linear error-block code if and only if $\lambda_{\min}=1$.
\end{Corollary}

\begin{Remark}
Irredundant mixed orthogonal arrays $\mathrm{IrMOA}(M,s,(q^{n_1},\dots,q^{n_s}),t)$ with $t=\lfloor s/2\rfloor$ and $\lambda_{\min}=1$ construct AME states of minimal support in heterogeneous systems (see Introduction). Theorem \ref{cor:extreme_cases} says that such MDS $\mathrm{IrMOA}(M,s,q^{n},t)$ with $t=\lfloor s/2\rfloor$ can cnstruct AME states of minimal support in heterogeneous systems. Or equivalently corollary says such   MDS $\mathbb F_q$-linear error-block code can construct AME states of minimal support in heterogeneous systems.
\end{Remark}

%=====================================================================

\subsection{$\mathbb F_q$-linear irredundant mixed orthogonal arrays}

Let $C \subseteq \mathbb{F}_q^n$ be an $\mathbb{F}_q$-linear error-block code of dimension $k$ and type $\pi=[n_1]\dots[n_s]$, where $n=\sum_{i=1}^s n_i$. It follows from Theorem~\ref{th:so} that the preimage $\rho^{-1}(C)$ is an $\mathbb{F}_q$-linear mixed orthogonal array
\begin{equation}\label{ne1}
\mathrm{MOA}(q^k, s, (q^{n_1}, \dots, q^{n_s}), t_1),
\end{equation}
where $t_1 = d_\pi(C^{\perp})-1$. Similarly, the preimage of the Euclidean dual $\rho^{-1}(C^\perp)$ is an
\begin{equation}\label{ne2}
\mathrm{MOA}(q^{n-k}, s, (q^{n_1}, \dots, q^{n_s}), t_2),
\end{equation}
with strength $t_2 = d_\pi(C)-1$.

The following theorem provides a method for constructing $\mathbb{F}_q$-linear IrMOAs from $\mathbb{F}_q$-linear error-block codes.

\begin{Theorem}\label{latest}
Let $C \subseteq \mathbb{F}_q^n$ be an $\mathbb{F}_q$-linear error-block code of type $\pi$ with parameters $[n, k, d_{\pi}(C)]$ and Euclidean dual $C^\perp$. Then $n \ge 2\min\{t_1, t_2\}$, and the following statements hold:
\begin{enumerate}
    \item \label{1} The array $\rho^{-1}(C)$ is an $\mathbb{F}_q$-linear irredundant MOA if and only if $d_\pi(C) \ge d_\pi(C^{\perp})$.
    \item Suppose either $q$ is even or both $n$ and $q$ are odd. Then:
    \begin{enumerate}
        \item \label{2a} The  array $\rho^{-1}(C^\perp)$ is an $\mathbb{F}_q$-linear irredundant MOA if and only if $d_\pi(C^{\perp}) \ge d_\pi(C)$.
        \item \label{2b} Both $\rho^{-1}(C)$ and $\rho^{-1}(C^\perp)$ are $\mathbb{F}_q$-linear irredundant MOAs if and only if $d_\pi(C) = d_\pi(C^{\perp})$.
    \end{enumerate}
\end{enumerate}
\end{Theorem}

\begin{proof}
By Corollary \ref{sari}, we obtain $n \ge t_1+t_2$, which implies that $n \ge 2\min\{t_1, t_2\}$.
\begin{enumerate}
    \item  By Theorem~\ref{th:so}, $\rho^{-1}(C)$ has strength $t_1 = d_\pi(C^{\perp})-1$. According to Lemma~\ref{m88}, $\rho^{-1}(C)$ is irredundant if and only if $d_H(\rho^{-1}(C)) \ge t_1 + 1 = d_\pi(C^{\perp})$. By the metric preservation property in Theorem~\ref{ma15}, $d_H(\rho^{-1}(C)) = d_\pi(C)$, which yields the result.
    \item  Under the stated conditions on $q$ and $n$, Theorem~\ref{ma15} implies that $\rho$ is dual-preserving, so $(\rho^{-1}(C))^{\perp_{\mathrm{Tr}}} = \rho^{-1}(C^\perp)$. Applying Lemma~\ref{m88} to the array $\rho^{-1}(C^\perp)$ of  strength $t_2 = d_\pi(C)-1$, $\rho^{-1}(C^\perp)$ is irredundant if and only if $d_H(\rho^{-1}(C^\perp)) \ge t_2 + 1 = d_\pi(C)$. Again by Theorem~\ref{ma15}, $d_H(\rho^{-1}(C^\perp)) = d_\pi(C^\perp)$, which proves \ref{2a}. Statement \ref{2b} follows immediately from the conjunction of \ref{1} and \ref{2a}.
\end{enumerate}
\end{proof}

%\begin{Example}
%Let $q=2$ and consider the partition $\pi =[2][1][1]$ of $n=4$ into $s=3$ blocks. Define the code $C \subseteq \mathbb{F}_2^4$ via the generator matrix $G$ and  the parity-check matrix $H$ given by
%$$
%G=
%\begin{bmatrix}
%1 & 0\mid  1\mid  0\\
%0 & 1 \mid  0 \mid  1
%\end{bmatrix}, \quad H \begin{bmatrix}
%1 & 0\mid  1\mid  0\\
%0 & 1 \mid  0 \mid  1
%\end{bmatrix}.
%$$
%It follows that $C^\perp=C$, and $d_\pi(C)=2=d_\pi(C^\perp)$. By Theorem \ref{latest}, the preimage $\rho^{-1}(C)$ is an $\mathbb{F}_2$-linear $\mathrm{IrMOA}(4, 3, (2^2, 2, 2), 1)$ of the form
%$$
%\begin{array}{ccc}
%0 & 0 & 0\\
%1 & 1 & 0\\
%\alpha & 0 & 1\\
%1+\alpha & 1 & 1
%\end{array}.
%$$
%\end{Example}

\begin{Example}
Let $q=2$ and consider the partition $\pi=[2][2][2][2][1][1]$ of $n=10$ into $s=6$ blocks. Define the code $C\subseteq \mathbb{F}_2^{10}$ via the generator matrix $G$ and the parity-check matrix $H$ given by
$$
G=
\begin{bmatrix}
0&0 \mid 1&1 \mid 0&1 \mid 0&1 \mid 1 \mid 1\\
0&1 \mid 0&0 \mid 0&1 \mid 1&1 \mid 1 \mid 0\\
1&0 \mid 0&1 \mid 0&1 \mid 1&0 \mid 0 \mid 0\\
0&0 \mid 0&1 \mid 1&1 \mid 1&1 \mid 0 \mid 1
\end{bmatrix},\quad
H=
\begin{bmatrix}
1&0 \mid 1&1 \mid 1&0 \mid 0&0 \mid 0 \mid 0\\
0&1 \mid 0&1 \mid 0&1 \mid 0&0 \mid 0 \mid 0\\
0&1 \mid 1&1 \mid 0&0 \mid 1&0 \mid 0 \mid 0\\
1&1 \mid 0&1 \mid 0&0 \mid 0&1 \mid 0 \mid 0\\
0&1 \mid 1&0 \mid 0&0 \mid 0&0 \mid 1 \mid 0\\
1&0 \mid 0&1 \mid 0&0 \mid 0&0 \mid 0 \mid 1
\end{bmatrix}.
$$
%The codewords of $C$ are 
%$$
%\begin{aligned}
%C=\{&
%00 \mid 00 \mid 00 \mid 00 \mid 0 \mid 0,\;
%00 \mid 01 \mid 11 \mid 11 \mid 0 \mid 1,\;
%10 \mid 01 \mid 01 \mid 10 \mid 0 \mid 0,\;
%10 \mid 00 \mid 10 \mid 01 \mid 0 \mid 1,\\
%&
%01 \mid 00 \mid 01 \mid 11 \mid 1 \mid 0,\;
%01 \mid 01 \mid 10 \mid 00 \mid 1 \mid 1,\;
%11 \mid 01 \mid 00 \mid 01 \mid 1 \mid 0,\;
%11 \mid 00 \mid 11 \mid 10 \mid 1 \mid 1,\\
%&
%00 \mid 11 \mid 01 \mid 01 \mid 1 \mid 1,\;
%00 \mid 10 \mid 10 \mid 10 \mid 1 \mid 0,\;
%10 \mid 10 \mid 00 \mid 11 \mid 1 \mid 1,\;
%10 \mid 11 \mid 11 \mid 00 \mid 1 \mid 0,\\
%&
%01 \mid 11 \mid 00 \mid 10 \mid 0 \mid 1,\;
%01 \mid 10 \mid 11 \mid 01 \mid 0 \mid 0,\;
%11 \mid 10 \mid 01 \mid 00 \mid 0 \mid 1,\;
%11 \mid 11 \mid 10 \mid 11 \mid 0 \mid 0
%\}.
%\end{aligned}
%$$
One can check that $d_\pi(C)=4$ and $d_\pi(C^\perp)=3$. Hence, Theorem~\ref{latest} implies that the preimage $\rho^{-1}(C)$ is an $\mathbb{F}_2$-linear $\mathrm{IrMOA}\bigl(16,6,(2^2,2^2,2^2,2^2,2,2),2\bigr)$ of the form 
$$
\begin{array}{cccccc}
0 & 0 & 0 & 0 & 0 & 0\\
0 & \alpha & 1+\alpha & 1+\alpha & 0 & 1\\
1 & \alpha & \alpha & 1 & 0 & 0\\
1 & 0 & 1 & \alpha & 0 & 1\\
\alpha & 0 & \alpha & 1+\alpha & 1 & 0\\
\alpha & \alpha & 1 & 0 & 1 & 1\\
1+\alpha & \alpha & 0 & \alpha & 1 & 0\\
1+\alpha & 0 & 1+\alpha & 1 & 1 & 1\\
0 & 1+\alpha & \alpha & \alpha & 1 & 1\\
0 & 1 & 1 & 1 & 1 & 0\\
1 & 1 & 0 & 1+\alpha & 1 & 1\\
1 & 1+\alpha & 1+\alpha & 0 & 1 & 0\\
\alpha & 1+\alpha & 0 & 1 & 0 & 1\\
\alpha & 1 & 1+\alpha & \alpha & 0 & 0\\
1+\alpha & 1 & \alpha & 0 & 0 & 1\\
1+\alpha & 1+\alpha & 1 & 1+\alpha & 0 & 0
\end{array}.
$$
For this IrMOA, we note that 
$$
\lambda_{i,j} = \frac{16}{2^{n_i+n_j}} = 
\begin{cases} 
1 & \text{if } i, j \in \{1, 2, 3, 4\}, \\ 
2 & \text{if one index is in } \{1, 2, 3, 4\} \text{ and the other in } \{5, 6\}, \\ 
4 & \text{if } i, j\in \{5, 6\}. 
\end{cases}
$$
\end{Example}

\medskip

{\bf Acknowledgments.} The research of Maryam Bajalan and Peter Boyvalenkov is supported, in part, by Ministry of Education and Science of Bulgaria under Grant no. DO1-98/26.06.2025 “National Centre for High-Performance and Distributed Computing”. The research of Ferruh \"{O}zbudak is supported by T\"{U}B\.{I}TAK under Grant 223N065.

%%%%%%%%%%%%%%%%%%%%%%%%%%%%%%%%%%%%%%%%%%%%%%%%%%%%%%%%%%%%%%%%%%%%%%%%%%%%%%%%%%%%%%%%%%%%%%%%%%%%%%%%%%%%%%%%%%%%

\end{document}